\documentclass[letterpaper, 10pt, doublecolumn, conference]{ieeeconf}

\IEEEoverridecommandlockouts                              
\title{\LARGE \bf
Employing Altruistic Vehicles at On-ramps to Improve the Social Traffic Conditions}

\usepackage{tabu}
\usepackage{epsfig} 
\usepackage{amssymb}  
\usepackage{amsmath}
\let\labelindent\relax
\usepackage{enumitem}
\usepackage{pifont}
\usepackage[dvipsnames]{xcolor}
\usepackage{tikz}
\usepackage{pgfplots}
\usepackage{algorithm2e}
\usepackage{varwidth}
\usepackage{soul}
\usepackage{verbatim}
\usepackage{cite}
\usepackage{booktabs}
\usetikzlibrary{arrows,shapes}
\usepackage{float}

\usepackage{graphics} 
\usepackage{amsmath} 
\usepackage{amssymb}  

\usepackage{amsthm}
\usepackage[scientific-notation=true]{siunitx}
\usepackage[colorlinks]{hyperref}
\usepackage{amsfonts}
\usepackage{tikz}
\usetikzlibrary{arrows.meta}
\usepackage{dsfont}

\usepackage{graphicx}
\usepackage{caption}
\usepackage{subcaption}
\usepackage{multirow}
\usepackage{setspace}  

\newtheorem{theorem}{Theorem}

\theoremstyle{definition}
\newtheorem{example}{Example}
\newtheorem{definition}{Definition}

\theoremstyle{remark}

\usepackage{placeins}

\definecolor{darkblue}{RGB}{0,101,204}
\definecolor{carorange}{RGB}{255,131,0}

\newcounter{tmp}

\author{Ruolin Li$^{1}$,   Philip N. Brown$^{2}$ and Roberto Horowitz$^{1}$
\thanks{$^{1}${R. Li and R. Horowitz are with the Department of Mechanical Engineering, University of California, Berkeley, CA, USA.
	{\tt\small ruolin\_li@berkeley.edu},
{\tt\small horowitz@me.berkeley.edu.}}
}
\thanks{$^{2}$P. N. Brown is with the Department of Computer Science, University of Colorado Colorado Springs, USA.
       {\tt\small philip.brown@uccs.edu.}}
}

\begin{document}

\maketitle


\thispagestyle{empty}
\pagestyle{empty}

\begin{abstract}
Highway on-ramps are regarded as typical bottlenecks in transportation networks. In previous work, mainline vehicles' selfish lane choice behavior at on-ramps is studied and regarded as one cause leading to on-ramp inefficiency. When on-ramp vehicles plan to merge into the mainline of the highway, mainline vehicles choose to either stay steadfast on the current lane or bypass the merging area by switching to a neighboring lane farther from the on-ramp. Selfish vehicles make the decisions to minimize their own travel delay, which compromises the efficiency of the whole on-ramp. Results in previous work have shown that, if we can encourage a proper portion of mainline vehicles to bypass rather than to stay steadfast, the social traffic conditions can be improved. In this work, we consider employing a proportion of altruistic vehicles among the selfish mainline vehicles to improve the efficiency of the on-ramps. The altruistic vehicles are individual optimizers, making decisions whether to stay steadfast or bypass to minimize their own altruistic cost, which is a weighted sum of the travel delay and their negative impact on other vehicles. We first consider the ideal case that altruistic costs can be perfectly measured by altruistic vehicles. We give the conditions for the proportion of altruistic vehicles and the weight configuration of the altruistic costs, under which the social delay can be decreased or reach the optimal. Subsequently, we consider the impact of uncertainty in the measurement of altruistic costs and we give the optimal weight configuration for altruistic vehicles which minimizes the worst case social delay under such uncertainty.

\end{abstract}

\section{Introduction}\label{intro}
Vehicles' selfish choice behavior on the transportation networks is regarded as one important cause of traffic inefficiency. In~\cite{pigou2013economics,roughgarden2002bad}, it is shown that the social delay may deteriorate under vehicles' selfish routing behavior. In~\cite{negar2018bypassing,li2019extended,li2020game}, vehicles' selfish lane--changing behavior at various traffic diverges and merges are studied and are shown to compromise the social traffic conditions. Highway on--ramps are one of the typical bottlenecks which are severely affected by vehicles' selfish lane--changing behaviors. Vehicles traveling through the on--ramp lane plan to merge into the highway, which stimulates various lane--changing behavior of the vehicles in the on--ramp area. The lane--changing behaviors' negative impact on the traffic conditions at on--ramps are studied in~\cite{cassidy1990proposed,daganzo2002merge,coifman2005lane}.

Researchers have resorted to various methods to constrain or guide vehicles' selfish behaviors to improve the social traffic conditions. In~\cite{cole2003pricing,fleischer2004tolls,mehr2019pricing,ferguson2019utilizing,brown2020can}, the impact of the toll pricing on vehicles is studied in vehicles' selfish routing scenario. In~\cite{beckmann1955studies,sandholm2002evolutionary}, a specific kind of toll called the marginal cost is proved to be capable of optimizing the social delay when vehicles are selfish routing. Recently, autonomous vehicles, which are more controllable than human--driven vehicles, are also increasingly studied on transportation networks. In~\cite{lioris2017platoons,mehr2019will,smith2020improving,li2020impact}, autonomous vehicles are shown to be capable of improving the social traffic conditions by preserving a shorter headway than human--driven vehicles. In~\cite{biyik2018altruistic}, autonomous vehicles are regarded as altruistic vehicles, which are willing to take routes with longer delay than the quickest route. 

In previous work~\cite{li2020game}, we considered a highway on--ramp with two lanes in the mainline as shown in Figure~\ref{fig:diverge_basic}. Vehicles traveling along the on--ramp lane 0 are trying to merge with mainline vehicles on lane 1. To minimize their own travel delay, selfish mainline vehicles along lane 1 choose between two options. The first option is to stay \textit{steadfast} on the current lane 1 and merge with the on--ramp vehicles. The second option is to switch to the neighboring lane 2 and \textit{bypass} the merging with on--ramp vehicles. We then modeled the aggregate choice behavior of selfish mainline vehicles as a Wardrop equilibrium~\cite{wardrop1952some}. The results from simulations show that the selfish lane choice behavior worsens the social delay and if properly more mainline vehicles bypass instead of staying steadfast on the current lane, then the social traffic conditions can be improved.

In this work, building on~\cite{li2020game}, we employ a proportion of altruistic vehicles among the selfish mainline vehicles to improve the social traffic conditions. \textit{Selfish} mainline vehicles choose to stay steadfast or bypass to minimize their own travel delay, whereas \textit{altruistic} vehicles make the decision to stay steadfast or bypass to minimize their own altruistic cost, which is a weighted average of the travel delay and the marginal cost~\cite{beckmann1955studies,sandholm2002evolutionary}. Altruistic vehicles are individual optimizers which require local delay and cost information but no centralized coordination. With the presence of autonomous vehicles, it is envisioned that connected and autonomous vehicles can be employed as altruistic vehicles perceiving and minimizing configured altruistic costs. The weight configuration of the altruistic costs indicates how altruistic vehicles are.
Naturally, the first question arises: will the altruism improve the social traffic conditions in the on--ramp lane choice scenario? To answer this question, we consider when the altruistic costs are perfectly measured by altruistic vehicles, and we find the conditions under which the altruism helps to decrease or optimize the social delay. The conditions indicate altruism always improves the traffic conditions when altruistic vehicles are abundant.
However, for the scenarios when altruistic vehicles only have inaccurate estimates of the altruistic costs, how altruistic should vehicles be?
We then give the optimal weight configuration for altruistic vehicles that minimizes the worst case social delay under such uncertainty.


\section{The Model}\label{sec:model}

\begin{figure}
\centering
\begin{tikzpicture}[scale=0.4]
  \fill[gray!30] (0,0)-- (0,12) 
  --(6,12) --(6,10) -- (9,7)  
  -- (9,4) -- (11,0)
  -- (8,0)-- (6,4)-- (6,0);

  \draw[very thick] (0,0)-- (0,12);
  \draw[very thick] (6,12) --(6,10) -- (9,7)  
  -- (9,4) -- (11,0);
  \draw[very thick] (8,0)-- (6,4)-- (6,0);

  \draw[very thick, dash pattern=on 5pt off 5pt, white] (3,0) -- (3,12);

    \node at (9.5,0.5) {0}; 
    \node at (4.5,0.5) {1};  
    \node at (1.5,0.5) {2};

  \draw[very thick, dashed, darkblue] (4.5,1) -- (4.5,12);
  \draw[very thick, dashed, carorange] (4.5,1)-- (4.5,3)  .. controls (3,4) and (2,8) .. (1.5,12);
  
\node[darkblue,scale=1] at (5.5,6) {$\hat{x}^s_1$};
  \node[carorange,scale=1] at (1.5,6) {$\hat{x}^b_1$};

\end{tikzpicture}
\caption{Problem setting: mainline vehicles on lane 1 (both selfish and altruistic) choose to stay steadfast on lane 1 or bypass the merging with on--ramp vehicles.}
\label{fig:diverge_basic}
\end{figure}
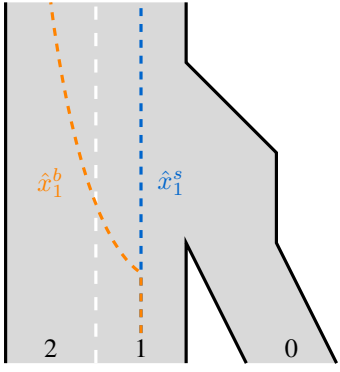


Let $I=\{0,1,2\}$ be the lane index set for the highway on--ramp in Figure~\ref{fig:diverge_basic},
where lane 0 is the on--ramp lane, lane 1 is the outermost lane in the mainline and vehicles on lane 1 make decisions to stay steadfast or bypass.
For lane $i\in\{0,2\}$, we let $n_i$ be the normalized flow on lane $i$, which indicates the relative magnitude of flow on lane $i$ among the neighboring flows of lane 1. Note that $n_0+n_2=1$. We consider $n_0$ and $n_2$ to be known and static. We then collect them in the neighboring flow configuration vector $\mathbf{N}:=(n_0,n_2)$. For lane 1,  let $x_1^{\rm s}$ (resp., $\Tilde{x}_1^{\rm s}$) represent the proportion flow of selfish (resp., altruistic) steadfast vehicles on lane 1, and let $x_1^{\rm b}$ (resp., $\Tilde{x}_1^{\rm b}$) represent the proportion flow of selfish (resp., altruistic) bypassing vehicles on lane 1. Now we collect the proportion flows on lane 1 in the flow distribution vector $\mathbf{x}:= (x_1^{\rm s}, x_1^{\rm b},\Tilde{x}_1^{\rm s}, \Tilde{x}_1^{\rm b})$. Let $\alpha\geq 0$ be the altruistic ratio, i.e., the proportion of altruistic vehicles among all the vehicles. A flow distribution vector is feasible if and only if
\begin{align}
    &x_1^{\rm s} + x_1^{\rm b} =1-\alpha,\label{eq:1}\\
    &\Tilde{x}_1^{\rm s} + \Tilde{x}_1^{\rm b} =\alpha,\label{eq:2}\\
    &x_1^{\rm s}\geq 0, x_1^{\rm b}\geq 0, \Tilde{x}_1^{\rm s}\geq 0, \Tilde{x}_1^{\rm b}\geq 0.
\end{align}

Selfish vehicles will only choose the option that can minimize their own travel delay. For selfish vehicles, we employ the delay models that are calibrated and validated in our previous work~\cite{li2020game}. Notice that the selfish flow proportions in the original delay models have to be replaced by the total flow proportions including both selfish and altruistic vehicles choosing the same option. For simplicity of future reference, we let $\hat{x}^{\rm s}_1:=x_1^{\rm s}+\Tilde{x}_1^{\rm s}$ and $\hat{x}^{\rm b}_1:=x_1^{\rm b}+\Tilde{x}_1^{\rm b}$ be the total proportion of steadfast and bypassing vehicles. Note that
\begin{align}\label{eq:5}
    \hat{x}^{\rm s}_1+\hat{x}^{\rm b}_1=1.
\end{align}
Let $J_1^{\rm s}$ denote the travel delay experienced by selfish steadfast vehicles, and $J_1^{\rm b}$ denote the delay experienced by selfish bypassing vehicles. We have
\begin{align}
J_1^{\rm s}(\mathbf{x}) &= C_1^{\rm t} \mu(\hat{x}_1^{\rm s}+n_0)+C_1^{\rm m} \hat{x}_1^{\rm s} n_0,\label{eq:J_1^s_basic}\\
J_1^{\rm b}(\mathbf{x}) &= C_2^{\rm t} \left(  \gamma \hat{x}_1^{\rm b}+n_2  \right) + C_2^{\rm m}  \hat{x}_1^{\rm b} n_2.\label{eq:J_1^b_basic}
\end{align}
The cost coefficients are collected in the cost coefficient vector $\mathbf{C} := (C_i^{\rm t},C_i^{\rm m},\mu,\gamma: i \in \{1,2\})$. The coefficients are all non--negative constants for an on--ramp, which need to be calibrated for each on--ramp. The detailed explanation of the coefficients can be seen in~\cite{li2020game}. For simplicity of reference, we rewrite the travel delay models for selfish vehicles:
\begin{align}
J_1^{\rm s}(\mathbf{x}) &= K^{\rm s}\hat{x}_1^{\rm s}+B^{\rm s},\label{eq:J_1^s_basic_2}\\
J_1^{\rm b}(\mathbf{x}) &= K^{\rm b} \hat{x}_1^{\rm b} +B^{\rm b},\label{eq:J_1^b_basic}
\end{align}
where $K^{\rm s}:=C_1^{\rm t} \mu+C_1^{\rm m}  n_0$, $B^{\rm s}:=C_1^{\rm t}  \mu n_0$, $K^{\rm b}:=C_2^{\rm t} \gamma+C_2^{\rm m}  n_2$ and $B^{\rm b}:=C_2^{\rm t}  n_2$ are all non--negative constants for an on--ramp with a given neighboring flow configuration $\mathbf{N}$. Essentially, we care about the social traffic conditions, i.e., the social delay. We also employ the delay models proposed in~\cite{li2020game} for vehicles on lane 2 and on--ramp vehicles. Let $J_0$ denote the delay experienced by on--ramp vehicles and $J_2$ denote the delay experienced by vehicles on lane 2. We have
\begin{align}
J_0(\mathbf{x}) &= K^{\rm s}\hat{x}_1^{\rm s}+B^{\rm s},\label{eq:J_0}\\
J_2(\mathbf{x}) &= K_2 \hat{x}_1^{\rm b} +B^{\rm b},\label{eq:J_2}
\end{align}
where $K_2:=C_2^{\rm t}+C_2^{\rm m} n_2$ is also a  non--negative constant for an on--ramp with a given neighboring flow configuration. Notice that vehicles choosing the same option experience the same travel delay no matter they are selfish or altruistic. Therefore, the social delay can be expressed as
\begin{align}\label{eq:social_cost}
   J_{\rm soc}(\mathbf{x})=\hat{x}_1^{\rm s} J_1^{\rm s}(\mathbf{x}) + \hat{x}_1^{\rm b} J_1^{\rm b}(\mathbf{x})+ n_0 J_0(\mathbf{x})+n_2 J_2(\mathbf{x}),
\end{align}
which is a convex function of $\mathbf{x}$ when $\mathbf{N}$ and $\mathbf{C}$ are given. Note that by Equation~\eqref{eq:5}, we always have $\hat{x}_1^{\rm s}=1-\hat{x}_1^{\rm b}$. In the later proofs, we tend to use $\hat{x}_1^{\rm b}$ instead of $\mathbf{x}$ as the self--variable of the delay functions.
Let $J_{\rm opt}:=\underset{\hat{x}_1^{\rm b}\in[0,1]}{\text{min}}J_{\rm soc}(\mathbf{x})$ denote the optimal social delay. Ideally, for an on--ramp with a given neighboring flow configuration, we aim to employ altruistic vehicles to decrease the social delay to its minimum.

Altruistic vehicles also choose to stay steadfast or bypass to minimize their own cost. However, to improve the social traffic conditions, they are configured to perceive an altruistic cost which is different from the travel delay. We propose the altruistic cost to be the weighted sum of the travel delay and the marginal cost (see~\cite{beckmann1955studies,sandholm2002evolutionary} for more details).
Let $\Tilde{J}_1^{\rm s}$ denote the altruistic cost for altruistic steadfast vehicles, and $\Tilde{J}_1^{\rm b}$ denote the cost for altruistic bypassing vehicles. We have 
\begin{align}
\Tilde{J}_1^{\rm s}(\mathbf{x}) &= (1-\beta)J_1^{\rm s}(\mathbf{x})+\beta\frac{\partial J_{\rm soc}(\mathbf{x})}{\partial \hat{x}_1^{\rm s}}\\
&=J_1^{\rm s}(\mathbf{x})+\beta K^{\rm s}(\hat{x}_1^{\rm s}+n_0),\label{eq:J_1^s_altruistic}
\end{align}
\begin{align}
\Tilde{J}_1^{\rm b}(\mathbf{x}) &=(1-\beta) J_1^{\rm b}(\mathbf{x})+\beta\frac{\partial J_{\rm soc}(\mathbf{x})}{\partial \hat{x}_1^{\rm b}}\\
&= J_1^{\rm b}(\mathbf{x})+\beta (K^{\rm b}\hat{x}_1^{\rm b}+ K_2n_2),\label{eq:J_1^b_altruistic}
\end{align}
where $0\leq \beta \leq1$ is the altruism level of the altruistic vehicles, which acts as the weight configuration of the altruistic costs. The altruism level is assigned to the altruistic vehicles by a central authority and can be interpreted as the propensity of altruistic vehicles to optimize the social delay. When $\beta=0$, altruistic vehicles behave exactly like selfish vehicles; when the altruistic ratio $\alpha=1$ and altruism level $\beta=1$, the resulting social delay is minimized. In this work, we consider the altruism level to be the same for all altruistic vehicles at the on--ramp. 
Note that altruistic vehicles are not explicitly coordinated, but are still individual optimizers evaluating the altruistic cost instead of the travel delay. Therefore, just as the selfish vehicles in~\cite{li2020game}, the equilibrium of the choice behavior of altruistic vehicles can be formulated as a Wardrop equilibrium~\cite{wardrop1952some}. Let a tuple $G = (\mathbf{N},\mathbf{C})$ be the full configuration of a highway on--ramp shown in Figure~\ref{fig:diverge_basic}, where $\mathbf{N}$ is the static neighboring flow configuration and $\mathbf{C}$ is the cost coefficient vector. The resulting choice equilibrium of the mixed selfish and altruistic vehicles is then defined as below.
\begin{definition}\label{def:wdp_basic}
For a given on--ramp configuration $G = (\mathbf{N}, \mathbf{C})$, a flow distribution vector $\mathbf{x}$ is a choice equilibrium if and only if
\begin{subequations}\label{eq:eq_def}
    \begin{align}
    x_1^{\rm s} (J_1^{\rm s}(\mathbf{x}) - J_1^{\rm b}(\mathbf{x})) &\leq 0 ,\\
    x_1^{\rm b} (J_1^{\rm b}(\mathbf{x}) - J_1^{\rm s}(\mathbf{x})) &\leq 0,\\
    \Tilde{x}_1^{\rm s} (\Tilde{J}_1^s(\mathbf{x}) - \Tilde{J}_1^b(\mathbf{x})) &\leq 0 ,\\
    \Tilde{x}_1^{\rm b} (\Tilde{J}_1^b(\mathbf{x}) - \Tilde{J}_1^s(\mathbf{x})) &\leq 0.
    \end{align}
\end{subequations}
\end{definition}

\noindent The definition mathematically expresses that selfish vehicles only choose the option with their own minimized travel delay, whereas altruistic vehicles choose the option with their own minimized altruistic cost. 


\section{When no uncertainty exists}\label{sec:nouncertainty}

In this section, we consider the scenarios when no uncertainty lies in the measurement of travel delay and altruistic costs. We aim to analyze the impact of the altruistic ratio $\alpha$ and the altruism level $\beta$ on the resulting social delay.

In the rest of the paper, we only discuss a certain meaningful set of on--ramp configurations. As discussed in~\cite{li2020game}, when all vehicles are selfish, properly encouraging the bypassing behavior of mainline vehicles can improve the social delay. Thus in this work, we only focus on the on--ramp configurations where selfish vehicles bypass less than the socially optimal scenario. Moreover, when all vehicles are selfish, at the choice equilibrium, if all the selfish vehicles are bypassing vehicles, we cannot decrease the social delay anymore by letting more vehicles bypass even with altruistic vehicles; if all the selfish vehicles choose to stay steadfast, as discussed in~\cite{li2020game}, for any $x_1^{\rm b}\in(0,1]$, we have $J_1^{\rm s}(x_1^{\rm b})< J_1^{\rm b}(x_1^{\rm b})$. Naturally, when there are altruistic vehicles, for any $\hat{x}_1^{\rm b}\in(0,1]$, we have $J_1^{\rm s}(\hat{x}_1^{\rm b})< J_1^{\rm b}(\hat{x}_1^{\rm b})$. Therefore, $x_1^{\rm b}=0$ for any altruistic ratio and any altruism level, and the choice equilibrium is only dependent on altruistic vehicles' choices. The choice equilibrium of a single class of vehicles is then much the same as discussed in~\cite{li2020game}. Furthermore, consider the scenario when all vehicles are bypassing vehicles at the socially optimal equilibrium, then to optimize the social delay, we have to make sure all altruistic vehicles always choose bypassing, i.e., $\Tilde{J}_1^{\rm s}(\hat{x}_1^{\rm b})> \Tilde{J}_1^{\rm b}(\hat{x}_1^{\rm b})$, for any $\hat{x}_1^{\rm b}\in[0,1)$, thus the problem totally depends on selfish vehicles' choice, thus becomes again very similar to what is discussed in~\cite{li2020game}.

Therefore, in the rest of the paper, we only consider the on--ramp configurations where we are able to employ altruistic vehicles to move the less efficient interior equilibrium when all vehicles are selfish to a socially optimal interior equilibrium with more bypassing vehicles.

Let $\Phi$ denote the bypassing flow at the interior equilibrium when all vehicles are selfish, which satisfies $J_1^{\rm s}(\Phi)=J_1^{\rm b}(\Phi)$. Solving the equation, we have
\begin{align}
    \Phi=\frac{K^{\rm s}+B^{\rm s}-B^{\rm b}}{K^{\rm s}+K^{\rm b}}.
\end{align}
Note that $\Phi$ is only dependent on the on--ramp configuration. For a specific on--ramp configuration, $\Phi$ is a constant. Moreover, recalling Equation~\eqref{eq:social_cost}, the social delay $J_{\rm soc}(\hat{x}_1^{\rm b})$is a convex quadratic function of $\hat{x}_1^{\rm b}$. Let $\Delta$ be the unique global minimization point of the quadratic function, i.e., $J_{\rm soc}(\Delta)=\underset{\hat{x}_1^{\rm b}\in \mathbb{R}}{\text{min}}J_{\rm soc}(\hat{x}_1^{\rm b})$. Note that $\Delta$ is only dependent on the on--ramp configuration. For a specific on--ramp configuration $G$, $\Delta$ is a constant regardless of any $\alpha$ or $\beta$ and we always have $J_{\rm soc}(\Delta)\leq J_{\rm opt}$. If the social delay decreases when the bypassing proportion increases for some on--ramp configurations, we must have $\Phi<\Delta$. Also, to ensure at the social optimum, not all vehicles are bypassing vehicles, we have $\Delta<1$.

Let $\mathcal{G}$ be the meaningful set of on--ramp configurations, and we then have
\begin{align}
    \mathcal{G}=\{G:0<\Phi<\Delta<1\}.
\end{align}

\begin{figure}
\centering
 \includegraphics[width=0.4\textwidth]{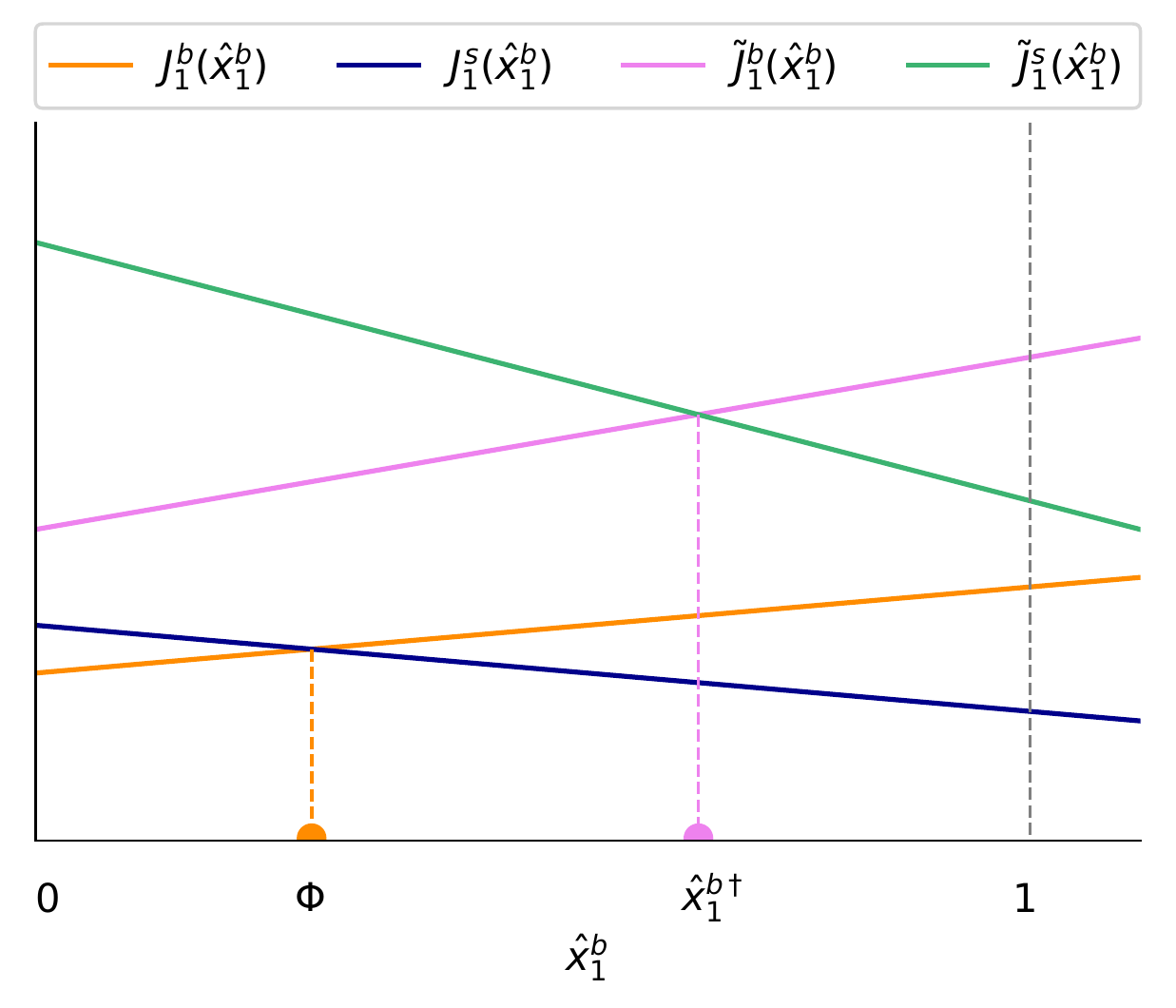}
    \caption{Sketch of the travel delay and altruistic cost functions, where $\Phi$ is indicated by the yellow dot and $\hat{x}_1^{\rm b\dagger}$ is indicated by the pink dot.}\label{fig:JJ}
\end{figure}

\noindent We then are ready to give the first core result in this work, which establishes the conditions for the altruistic vehicles' configurations to improve the social conditions or to reach the optimal social conditions.
\begin{theorem}\label{thm:altruistic ratio start}
For a given on--ramp configuration $G = (\mathbf{N}, \mathbf{C})\in \mathcal{G}$ with altruistic ratio $\alpha$ and altruism level $\beta$.
\begin{itemize}
    \item The social delay is decreased by altruistic vehicles , i.e., $J_{\rm soc}(\hat{x}_1^{\rm b})<J_{\rm soc}(\Phi)$, if and only if $\beta > 0$ and $\alpha\in \mathcal{A}_1$, where $\mathcal{A}_1:=(\Phi,1]$.
\item The social delay is optimized by altruistic vehicles, i.e., $J_{\rm soc}(\hat{x}_1^{\rm b})=J_{\rm opt}$, if and only if $\beta = 1$ and $\alpha\in \mathcal{A}_2$, where $\mathcal{A}_2:=\left[\Delta,1\right]$.
\end{itemize}
\end{theorem}

\begin{proof}
Let $\hat{x}_1^{\rm b\dagger}$ be the total bypassing proportion flow at the intersection of the altruistic steadfast and bypassing costs, i.e., $\Tilde{J}_1^{\rm s}(\hat{x}_1^{\rm b\dagger})=\Tilde{J}_1^{\rm b}(\hat{x}_1^{\rm b\dagger})$. Solving the equation, we have
\begin{align}\label{eq:betafunc}
    \hat{x}_1^{\rm b\dagger}=\frac{1-\beta}{1+\beta}\Phi+\frac{2\beta}{1+\beta}\Delta.
\end{align}
Notice that $\hat{x}_1^{\rm b\dagger}$ is a function of the altruism level $\beta$ for a certain on--ramp configuration. When $\beta=0$, the altruistic costs are exactly the same as the selfish delay. Therefore, altruistic vehicles behave the same as selfish vehicles. Manipulating altruistic vehicles would bring no change to the social delay. In the meaningful set $\mathcal{G}$, we always have $\Phi<\Delta$. When $\beta>0$, we have $\Phi<\hat{x}_1^{\rm b\dagger}$ and $\hat{x}_1^{\rm b\dagger}$ is always an increasing function of $\beta$. Note that $\hat{x}_1^{\rm b\dagger}$ can be seen as a weighted average of $\Phi$ and $\Delta$. We always have $\hat{x}_1^{\rm b\dagger} \leq \Delta$ and when $\beta=1$, we have $\hat{x}_1^{\rm b\dagger} = \Delta$. We then enumerate all the possible cases of the resulting equilibria when altruistic vehicles are involved, i.e., $\alpha>0$ and the altruism level $\beta>0$. See Figure~\ref{fig:JJ} which sketches the delays and altruistic costs, at the resulting equilibrium,
\begin{itemize}
    \item Case (a): if $\hat{x}_1^{\rm b}\in [0,\Phi)$, since the bypassing delay is smaller than the steadfast delay and the bypassing altruistic cost is smaller than the steadfast altruistic cost, all vehicles will choose bypassing, i.e., $\hat{x}_1^{\rm b}=1$. Since $\Phi<1$, the conclusion $\hat{x}_1^{\rm b}=1$ contradicts the assumption that $\hat{x}_1^{\rm b}<\Phi$. Therefore, the equilibrium cannot lie in this case.
    \item Case (b): if $\hat{x}_1^{\rm b} = \Phi$, the bypassing altruistic cost is smaller than the steadfast altruistic cost, thus all altruistic vehicles are bypassing vehicles, i.e., $\Tilde{x}_1^{\rm b}=\alpha$. Therefore, we have $x_1^{\rm b}=\Phi-\alpha$. The requirement for this case to happen is $x_1^{\rm b}\geq 0$, i.e. $\alpha\leq \Phi$.
    \item Case (c): if $\hat{x}_1^{\rm b}\in(\Phi,\hat{x}_1^{\rm b\dagger})$, the bypassing altruistic cost is smaller than the steadfast altruistic cost, thus all altruistic vehicles choose bypassing, i.e., $\Tilde{x}_1^{\rm b}=\alpha$. However, the bypassing delay is larger than the steadfast delay, thus all selfish vehicles will stay steadfast, i.e., $x_1^{\rm b}=0$. Therefore, we have $\hat{x}_1^b=\Tilde{x}_1^{\rm b}=\alpha$. Thus, the requirement of this case is $\Phi<\alpha<\hat{x}_1^{\rm b\dagger}$.
    \item Case (d): if $\hat{x}_1^{\rm b}= \hat{x}_1^{\rm b\dagger}$, the bypassing delay is larger than the steadfast delay, thus all selfish vehicles will stay steadfast, i.e., $x_1^{\rm b}=0$. Then we have $\hat{x}_1^b=\Tilde{x}_1^{\rm b}=\hat{x}_1^{\rm b\dagger}$. The requirement for this case is $\Tilde{x}_1^{\rm b}\leq \alpha$, i.e., $\hat{x}_1^{\rm b\dagger}\leq \alpha$.
    \item Case (e): if $\hat{x}_1^{\rm b}\in(\hat{x}_1^{\rm b\dagger},1]$, since bypassing delay is larger than the steadfast delay and the bypassing altruistic cost is larger than the steadfast altruistic cost, all vehicles must stay steadfast, i.e., $\hat{x}_1^{\rm b}=0$. Since $\hat{x}_1^{\rm b\dagger}>\Phi>0$, the equilibrium cannot lie in this case.
\end{itemize}

Consider the altruism level $\beta>0$ as fixed, and consider the altruistic ratio $\alpha$ as a variable. Recall that when $\alpha=0$, all vehicles are selfish, and at the equilibrium, $\hat{x}_1^b=\Phi$, we have the social delay as $J_{\rm soc}(\Phi)$. We then only need to consider case (b), (c) and (d). In case (b), the social delay remains the same as $J_{\rm soc}(\Phi)$, whereas in case (c) and (d), at the equilbrium, we have $\Phi<\hat{x}_1^{\rm b}\leq \Delta$. Recall that the social delay function~\eqref{eq:social_cost} is a convex quadratic function of $\hat{x}_1^{\rm b}$ with a minimization point $\Delta$, thus we have $J_{\rm soc}(\hat{x}_1^{\rm b})<J_{\rm soc}(\Phi)$. Since $0<\Delta<1$, then $J_{\rm opt}=J_{\rm soc}(\Delta)$, and therefore, the optimal social delay is only reached when $\hat{x}_1^{\rm b}=\hat{x}_1^{\rm b\dagger}=\Delta$, i.e., $\beta=1$ and the condition for case (c) is fulfilled.

In a nutshell, to decrease the social delay, we have to set $\beta>0$ and $\alpha>\Phi$. Otherwise, when $\beta=0$, altruistic vehicles act exactly like selfish vehicles, making no change; when $\alpha\leq\Phi$, the equilibrium is always in case (b) and we always have $\hat{x}_1^{\rm b}=\Phi$. The social delay remains the same as $J_{\rm soc}(\Phi)$. Since $0<\Delta<1$, we must have $J_{\rm opt}=J_{\rm soc}(\Delta)$. Therefore, to reach the optimal social delay, we have to let $\hat{x}_1^{\rm b}=\Delta$, which is only possible in case (d) when $\hat{x}_1^{\rm b\dagger}=\Delta$, i.e., $\beta=1$. See Figure~\ref{fig:social} as an example.

\end{proof}

Theorem~\ref{thm:altruistic ratio start} shows that with slightly altruistic vehicles, the social conditions can be improved as long as there are enough altruistic vehicles; however, to reach the optimal social traffic conditions, we have to employ enough purely altruistic vehicles. The altruism level of altruistic vehicles decides the best case of social delay we can do with abundant altruistic vehicles.

\begin{figure}
    \centering 
    \includegraphics[width=0.4\textwidth]{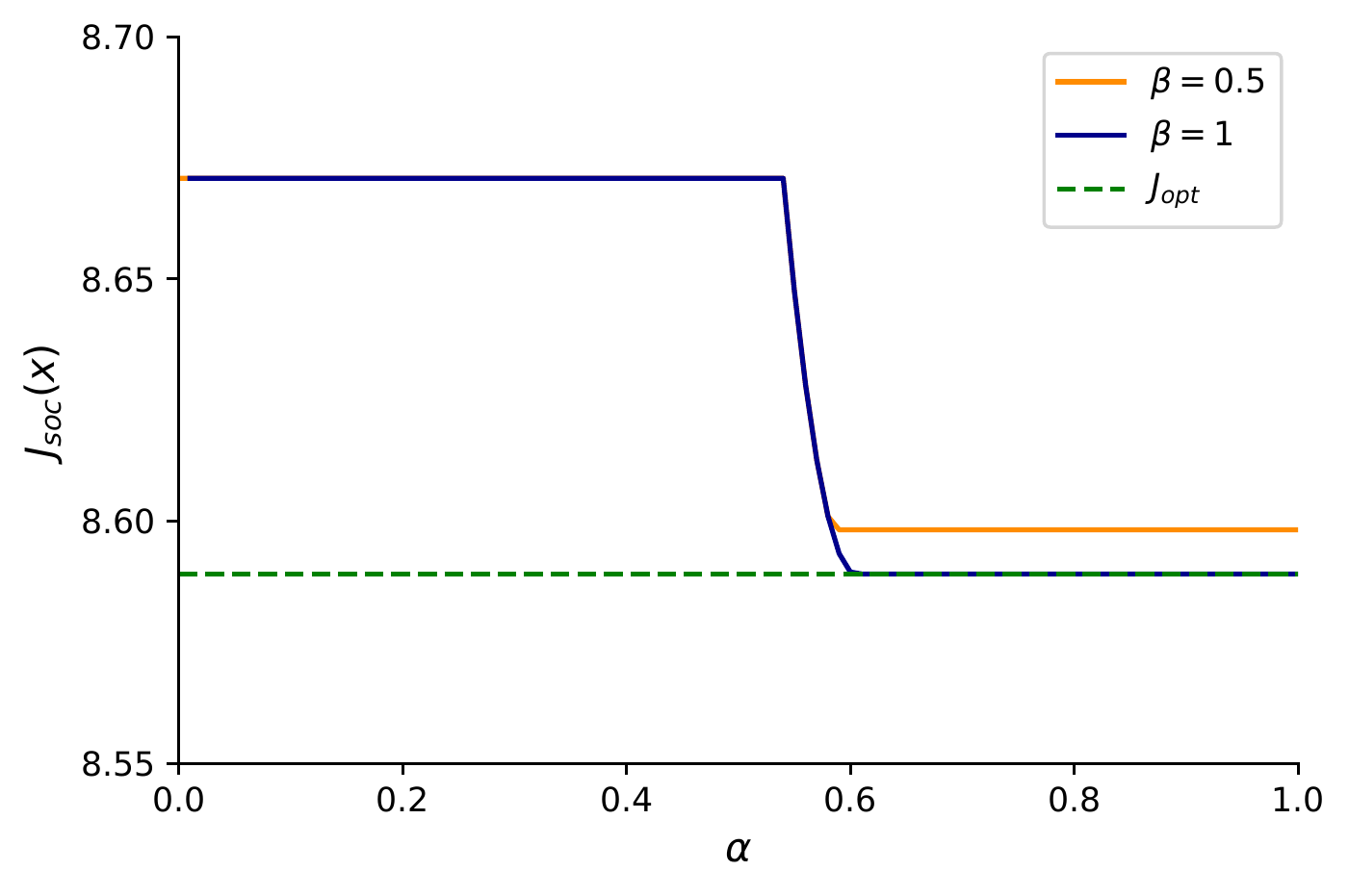}
    \caption{The social delay versus the altruistic ratio under different altruism levels. The on--ramp cost coefficients are
$C_1^t = C_2^t = 1,\ C_1^m = 21.3,\ C_2^m = 1,\ \mu=2.4,\ \gamma=8.6$ and the neighboring flow configuration is $n_0=0.37$. The on--ramp configuration lies in the meaningful set $\mathcal{G}$. As we can see, when altruistic vehicles are not abundant or the altruism level is less than 1, the social delay improvement is compromised. }\label{fig:social}
\end{figure}

\section{When uncertainty exists} \label{sec:uncertainty}

In reality, measuring the travel delay could be relatively easy whereas having an exact estimate of the altruistic part of the altruistic cost is hard. 
In this section, we assume that travel delay can be perfectly measured whereas some error $e\sim\mathcal{E}(e_{\rm L},e_{\rm U})$ is embedded in the altruistic costs, where $\mathcal{E}(e_{\rm L},e_{\rm U})$ is some probability distribution with a lower bound $e_{\rm L}> 0$ and an upper bound $e_{\rm U}>e_{\rm L}$. We assume all altruistic vehicles are affected by a homogeneous error, then the actual altruistic costs perceived by altruistic vehicles are
\begin{align}\label{eq:error}
\Tilde{J}_1^{\rm s}(\mathbf{x}) &= J_1^{\rm s}(\mathbf{x})+\beta e K^{\rm s}(\hat{x}_1^{\rm s}+n_0),\\
\Tilde{J}_1^{\rm b}(\mathbf{x}) &= J_1^{\rm b}(\mathbf{x})+\beta e (K^{\rm b}\hat{x}_1^{\rm b}+ K_2n_2).\label{eq:J_1^b_altruistic}
\end{align}

Without this uncertainty, we can easily use Theorem~\ref{thm:altruistic ratio start} to set $\beta=1$ and when altruistic vehicles are abundant, i.e., $\alpha\in\mathcal{A}_2$, the social delay reaches optimal. However, with this uncertainty, we may not optimize the social delay as we expected. Thus, we are interested in the worst case social delay for different choice of altruism levels with the presence of the uncertainty. In the end, we aim to find an optimal altruism level to minimize the worst case social delay with respect to the cost uncertainty and level of on--ramp configuration uncertainty. We first give the following definition of the generalized price of anarchy to characterize the worst case social delay in our problem setting.

\begin{definition}\label{def:poa}
For a given on--ramp configuration $G = (\mathbf{N}, \mathbf{C})\in \mathcal{G}$ with an altruism level $\beta$, the price of anarchy (PoA) is defined as
\begin{align}\label{eq:eq_def}
    PoA(G,\beta,\mathcal{E}):=\underset{e \sim\mathcal{E}(e_{\rm L},e_{\rm U})}{\text{sup}}\underset{\alpha\in \mathcal{A}_2}{\text{sup}}\frac{J_{\rm soc}(G,\beta,e,\alpha)}{J_{\rm opt}(G)}.
\end{align}
\end{definition}
\noindent Note that the design goal is to optimize the social delay, therefore, the price of anarchy only focuses on the altruistic ratios in the $\mathcal{A}_2$ range. Then we define the optimal altruism level that we are trying to find.

\begin{definition}\label{def:betastar}
For an on--ramp configuration $G\in \mathcal{G}$, the optimal altruism level $\beta^*$ satisfies
\begin{align}
    \beta^*=\text{arg}\ \underset{\beta\geq0}{\text{min}}\ PoA(G,\beta,\mathcal{E}).
\end{align}
\end{definition}

\noindent Let $\Pi:=\frac{1-\Phi}{2\Delta-\Phi-1}$. Note that in $\mathcal{G}$, we always have $1-\Phi>0$, thus $\Pi$ is nonzero. Let $\mathcal{G}_1:=\left\{G \in \mathcal{G}:0<\Pi<\sqrt{\frac{e_{\rm U}}{e_{\rm L}}}\right\}$ and $\mathcal{G}_2:=\{G \in \mathcal{G}\setminus\mathcal{G}_1\}$. We are ready to give the following theorem.

\begin{figure}
    \centering 
    \includegraphics[width=0.4\textwidth]{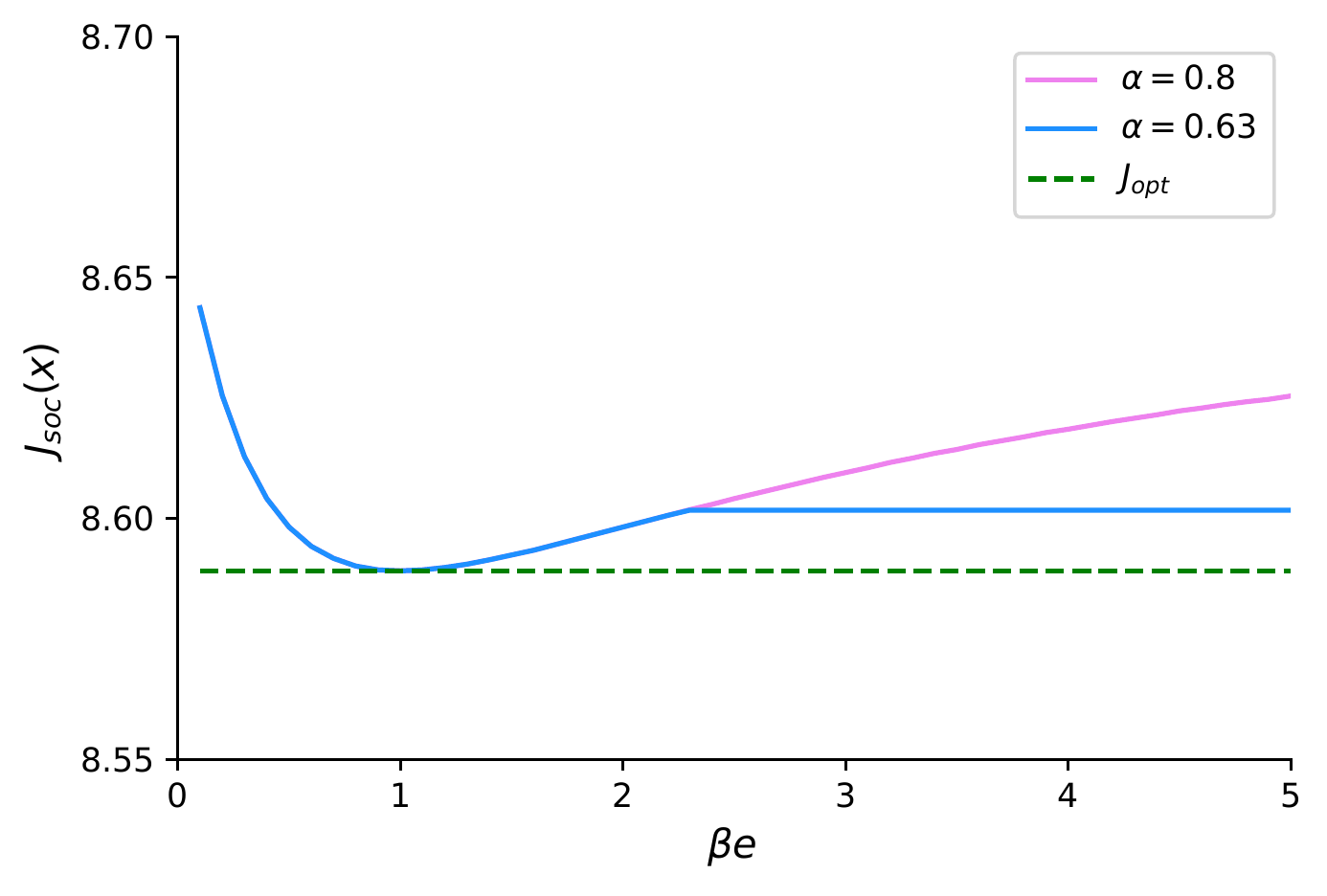}
    \caption{The social delay versus $\beta e$ under different altruistic ratios. The on--ramp cost coefficients are
$C_1^t = C_2^t = 1,\ C_1^m = 21.3,\ C_2^m = 1,\ \mu=2.4,\ \gamma=8.6$ and the neighboring flow configuration is $n_0=0.37$. The on--ramp configuration lies in the set $\mathcal{G}_2$. The worst case social delay happens on the pink curve. The optimal altruism level satisfies $\beta^*=\frac{1}{\sqrt{e_{\rm L}e_{\rm U}}}$.}\label{fig:social2}
\end{figure}

\begin{theorem}\label{thm:1}
For an on--ramp configuration $G\in\mathcal{G}_1$, the optimal altruism level $\beta^*=\frac{1}{e_{\rm L}\Pi}$; for an on--ramp configuration $G\in \mathcal{G}_2$, the optimal altruism level $\beta^*=\frac{1}{\sqrt{e_{\rm L}e_{\rm U}}}$.
\end{theorem}

\begin{proof}
In the choice model considering the uncertainty, we see the product of $e$ and $\beta$ as the ``effective'' altruism level which corresponds to the altruism level in the cost models considering no uncertainty. The cases of equilibria are then enumerated in the proof of Theorem~\ref{thm:altruistic ratio start}. Now consider the on--ramp configuration $G$, the altruistic ratio $\alpha$ and the error bounds $e_{\rm L}$ and $e_{\rm U}$ are fixed and see the product of $\beta e$ as a variable.
Since in the definition of PoA, $\alpha\in \mathcal{A}_2$, we must have $\alpha> \Phi$. Therefore, we only need to consider case (c) and (d). In a nutshell, when $\alpha\geq \hat{x}_1^{\rm b\dagger}$, at the equilibrium, $\hat{x}_1^{\rm b}=\hat{x}_1^{\rm b\dagger}$, and the social delay is calculated as $J_{\rm soc}(\hat{x}_1^{\rm b\dagger})$; when $\alpha<\hat{x}_1^{\rm b\dagger}$, at the equilibrium, $\hat{x}_1^{\rm b}=\alpha$, and the social delay is calculated as $J_{\rm soc}(\alpha)$. Recall that $\hat{x}_1^{\rm b\dagger}$ is a function of $\beta e$ (see Equation~\eqref{eq:betafunc}), we then may use the notation $\hat{x}_1^{\rm b\dagger}(\beta e)$. Now let us consider
\begin{itemize}
\item Case (A): $\Pi<0$. In this case, the range $\mathcal{A}_2$ can be divided into two ranges.
    \begin{itemize}
        \item Case (A.1): For any $\alpha\in \left[\Delta,2\Delta-\Phi \right)$, the resulting equilibrium changes from case (d) to case (c) when $\hat{x}_1^{\rm b\dagger}(\beta e)=\alpha$. Let $\Tilde{\beta}$ denote the transition point, where $\hat{x}_1^{\rm b\dagger}(\Tilde{\beta})=\alpha$. Solving the equation, we have $\Tilde{\beta}(\alpha)=\frac{\alpha-\Phi}{2\Delta-\Phi-\alpha}\geq 1$, which is an increasing function of $\alpha$. For $\beta e\leq \Tilde{\beta}$, the equilibrium is at $\hat{x}_1^{\rm b}=\hat{x}_1^{\rm b\dagger}$, and the social delay can be calculated as $J_{\rm soc}(\hat{x}_1^{\rm b\dagger})$. Since the social delay function is a convex quadratic function of $\hat{x}_1^{\rm b\dagger}$,  $\hat{x}_1^{\rm b\dagger}$ is an increasing function of $\beta e$, and $\hat{x}_1^{\rm b\dagger}(\beta e=1)=\Delta$, the social delay will first decrease, reaches optimal at $\beta e=1$ and then increase to $J_{\rm soc}(\alpha)$. For $\beta e> \Tilde{\beta}$, the equilibrium remains at $\hat{x}_1^{\rm b}=\alpha$, and the social delay remains the same at $J_{\rm soc}(\alpha)$. Note that when $\alpha$ increases, since we have $\alpha\geq \Delta$, $J_{\rm soc}(\alpha)$ increases.\\
        For clarity, Figure~\ref{fig:social2} gives an example of an on--ramp configuration satisfying $\Pi<0$. The blue curve ($\alpha=0.63$) corresponds to case (A.1).
        
        \item Case (A.2): When $\alpha\in \left[2\Delta-\Phi,1\right]$, for any $\beta e\geq 0$, the equilibrium is at $\hat{x}_1^{\rm b}=\hat{x}_1^{\rm b\dagger}$. The social delay can be calculated as $J_{\rm soc}(\hat{x}_1^{\rm b\dagger})$ and similarly to case (A.1), the social delay will first decrease, reaches optimal at $\beta e=1$ and then increase with $\beta e$ increasing.\\
        The pink curve ($\alpha=0.8$) in Figure~\ref{fig:social2} corresponds to case (A.2).
\end{itemize}
 The shape of the social delay function in case (A.2) can be seen as a whole curve without the stage after $\Tilde{\beta}$ in case (A.1). For the on--ramp configurations in case (A), considering all $\alpha\in \mathcal{A}_2$, the equilibrium with the worst case social delay happens at case (A.2). Notice that the function $J_{\rm soc}(\hat{x}_1^{\rm b\dagger})$ is a convex quadratic function of $\hat{x}_1^{\rm b\dagger}$ and $\hat{x}_1^{\rm b\dagger}$ is an increasing function of $\beta e$. Therefore, to minimize the PoA, we let 
 \begin{align}
     \hat{x}_1^{\rm b\dagger}(\beta^*e_{\rm U})-\Delta=\Delta-\hat{x}_1^{\rm b\dagger}(\beta^*e_{\rm L}).
 \end{align}
 Thus, we have $\beta^*=\frac{1}{\sqrt{e_{\rm L}e_{\rm U}}}$, and $J_{\rm soc}(\hat{x}_1^{\rm b\dagger}(\beta^*e_{\rm L}))=J_{\rm soc}(\hat{x}_1^{\rm b\dagger}(\beta^*e_{\rm U}))$.

\item case (B): $\Pi>0$. In this case, for any $\alpha\in \mathcal{A}_2$, similar to case (A.1), the resulting equilibrium changes from case (d) to case (c) when $\hat{x}_1^{\rm b\dagger}=\alpha$. As discussed in case (A.1), for the on--ramp configurations in case (B), considering all $\alpha\in \mathcal{A}_2$, the equilibrium with the worst case social delay happens when $\alpha=1$. Note that $\Pi=\Tilde{\beta}(\alpha=1)$, which is exactly the transition point when $\alpha=1$. Due to $\Delta<1$ in $\mathcal{G}$, we always have $\Pi>1$.  
The worst social delay function with respect to $\beta e$ is then a combination of the function $J_{\rm soc}(\hat{x}_1^{\rm b\dagger}(\beta e))$ and the constant stage $J_{\rm soc}(1)$ starting at $\beta e=\Pi$. We thus have to consider two scenarios. Letting $\beta_{gm}=\frac{1}{\sqrt{e_{\rm L}e_{\rm U}}}$, if $J_{\rm soc}(\hat{x}_1^{\rm b\dagger}(\beta_{gm} e_{\rm L}))=J_{\rm soc}(\hat{x}_1^{\rm b\dagger}(\beta_{gm} e_{\rm U}))$, then the same as case (A), we have $\beta^*=\beta_{gm}$; if $J_{\rm soc}(\hat{x}_1^{\rm b\dagger}(\beta_{gm} e_{\rm L}))>J_{\rm soc}(\hat{x}_1^{\rm b\dagger}(\beta_{gm} e_{\rm U}))$, then the constant stage has been reached at the upper bound of the error, i.e., $\beta_{gm} e_{\rm U}>\Pi$. Therefore, we could care less about the upper bound but more about the lower bound. To minimize the worst case social delay, instead of equalizing the social delay on the upper bound and the lower bound, we equalize the social delay of the lower bound and of the transition point $\Pi$. Therefore, we let 
\begin{align}
    \hat{x}_1^{\rm b\dagger}(\Pi)-\Delta=\Delta-\hat{x}_1^{\rm b\dagger}(\beta^*e_{\rm L}).
\end{align}Thus, we have $\beta^*=\frac{1}{e_{\rm L}\Pi}$.
\end{itemize}
Summarizing the cases, only when $\beta_{gm} e_{\rm U}>\Pi$ is satisfied in case (B), i.e., $0<\Pi<\sqrt{\frac{e_{\rm U}}{e_{\rm L}}}$, we have $\beta^*=\frac{1}{e_{\rm L}\Pi}$; otherwise, we always have $\beta^*=\frac{1}{\sqrt{e_{\rm L}e_{\rm U}}}$.
\end{proof}

According to Theorem~\ref{thm:1}, when measurements are imperfect, we can configure the altruistic vehicles with the optimal altruism level to minimize the worst case social delay under the uncertainty.

\section{conclusion} \label{sec:future}
In this work, we employ altruistic vehicles among the selfish mainline vehicles to improve the social traffic conditions of the on--ramps. We gave the conditions of the altruistic ratio and altruism level for altruistic vehicles to decrease or optimally decrease the social delay. Further, we assumed uncertainty in the altruistic cost measurements and we gave the optimal altruism level to configure altruistic vehicles which minimizes the worst case social delay under the uncertainty.

\section*{Acknowledgments}
This work was supported by the National Science Foundation under Grants CPS 1545116 and ECCS-2013779.


\bibliographystyle{IEEEtran}
\bibliography{onramp.bib}

\end{document}